\documentclass{llncs}
\usepackage[latin1]{inputenc}
\usepackage{amsmath,amsfonts}
\usepackage{enumitem}
\usepackage{tikz,tikz-qtree}
\usetikzlibrary{shapes,arrows}
\usepackage{epsfig}
\usepackage{pst-barcode}

\usepackage[linesnumbered,ruled,vlined]{algorithm2e}
\usepackage{epsfig}
\usepackage{amsmath,amssymb}
\newtheorem{thm}{Theorem}

\date{}

\def\bbe{{\mathbb E}}
\def\bbp{{\mathbb P}}

\def\itm#1{\par\parindent=20pt\noindent
   \hangindent\parindent\hbox to\parindent{#1\hss}\ignorespaces}
\def\dperp{\mathrel{\perp\kern-1em\perp}}

\makeatletter
\pgfmathdeclarefunction{square}{1}{%
\begingroup
 \pgfmathparse{#1*#1}
 \pgfmath@smuggleone\pgfmathresult%
\endgroup}

\title{Random-bit optimal uniform sampling for rooted planar trees with given sequence of degrees and Applications}

\author{O.Bodini \thanks{Supported by ANR Magnum project BLANC 0204 (France)} \and J. David \thanks{Supported by ANR Magnum project BLANC 0204 (France)} \and Ph. Marchal}

\institute{LIPN, Institut Galil\'ee, Universit\'e Paris 13, Villetaneuse (France)\\ LAGA, Institut Galil\'ee, Universit\'e Paris 13, Villetaneuse (France)}

\tikzstyle{block} = [draw, rectangle,minimum height=3em, minimum width=6em, text centered, text width=6em]
\tikzstyle{comment} = [minimum height=3em, minimum width=15em, text centered, text width=15em]

\begin{document}
\maketitle

\abstract{In this paper, we redesign and simplify an algorithm due to Remy et al. 
for the generation of rooted planar trees that satisfies a given partition of degrees.
%The interest of the paper resides in the simplicity of the proofs which could be explained to any undergrade student.
This new version is now optimal in terms of random bit complexity, up to a multiplicative constant. 
We then apply a natural process ``simulate-guess-and-proof'' to analyze the height of a random Motzkin in function of its frequency of unary nodes. When the number of unary nodes dominates, we prove some unconventional height phenomenon (i.e. outside the universal $\Theta(\sqrt{n} )$ behaviour.)}

\section{Introduction}

Trees are probably among the most studied objects in combinatorics,
computer science and probability. The literature on the subject is
abundant and covers many aspects (analysis of structural properties
such as height, profile, path length, number of patterns, but also
dynamic aspects such as Galton-Watson processes or random generation,
...) and use various techniques such as analytic combinatorics, graph
theory, probability,...

More particularly, in computer science, trees are a natural way to structure and manage
data, and as such, they are the basis of many crucial algorithms
(binary search trees, quad-trees, 2-3-4 trees, ...). In this article, we are
essentially interested in the random sampling of rooted planar trees.
This topic itself is also subject to a extensive study. To mention
only the best known algorithms, we can distinguish four approaches.
The first two of them are in fact more general, but can be applied
efficiently to the sampling of trees, the two others are ad hoc to
tree sampling:

\begin{enumerate}
\item
The random sampling by the recursive method \cite{FVCZ94} of generating a tree
from rules described with coefficients associated generating series \cite{DPT10},
\item The random generation under Boltzmann model that allows uniform
generation to approximate size from the evaluation of generating functions~\cite{DFLS04,BP10}.
\item The random generation by Galton-Watson processes based on the
dynamics of branching processes \cite{D12},
\item Samplers following Remy precepts~\cite{R85,ARS97,ARS97-2,BBJ13}.
\end{enumerate}

Concerning the generation of trees with a fixed degree sequence, the
reference algorithms are due to Alonso et al~\cite{ARS97}. However, the complete understanding of their approach
seems to us quite intricate. Moreover their approach is not optimal in terms of entropy (i.e. the minimum numbers of random bits necessary to draw an object uniformly as described in the famous Knuth-Yao paper~\cite{KY76}). 

In this article, we give two versions of an algorithm for drawing efficiently trees whose the degree sequence is given. 
Our first version is fast and easy to implement, and its description is simple and (we hope) natural.
It works, essentially like Alonso's algorithm, though we explicitly use the Lukasiewicz code of trees. 
Our second version only modifies the two first steps of the first algorithm. 
It is nearly optimal in terms of entropy because it uses only in average a linear number of random bits to draw a tree. 
Moreover, Lukasiewicz codes and a very elementary version of cyclic lemma allows us to give a simple proof of the theorem of Tutte~\cite{T64}
which gives under an explicit multinomial form the number of plane trees with a given partition of the degrees. 

From our sampler, we simulate various kind of trees. We focus our attention on unary-binary rooted planar trees (also called Motzkin trees) 
with a fixed frequency of unary nodes. 
In particular, we look for the variation of the height depending of the frequency of unary nodes. We can easily conjecture the  
nature of the variation. 

Our second contribution is to describe and prove the distribution of the height according to the number of unary nodes. 
The proof follows a probabilistic approach and uses in a central the notion of continuous random trees (CRT).
Even if the distribution of the height still follows a classical theta law, the expected value can leave the universal $\Theta(\sqrt{n} )$ behaviours.

The general framework used in this paper to describe trees is the analytic combinatorics even if we use some classical notion on word theory and a basis of probabilistic concepts in the second part of the paper. 
More specifically, we deal with the symbolic method to describe the bijection between Lukasiewicz words and trees. 
A {\em combinatorial class} is a set of discrete objects $\mathcal{O}$, provided with a (multidimensional) \emph{size function} $s: \mathcal{O} \rightarrow \mathbb{N}^d$ for some integer $d$, in such way that for every $\boldsymbol{n} \in \mathbb{N}^d$, the set of discrete objects of size $\boldsymbol{n}$, denoted by $\mathcal{O}_{\boldsymbol{n}} $, is finite. In the classical definition, the size is just scalar, but for our parametrized problem this extension is more convenient. For more details, see for instance \cite{FS09}. This approach is very well suited to the definition of trees. For instance, the class of binary trees $\mathcal{B}$ can be described by the following classical specification : $\mathcal{B}=\mathcal{Z}+\mathcal{Z}\mathcal{B}^2$.

In this framework, random sampling can be interpreted as follows. A {\em size uniform random generator }is an algorithm that generate discrete objects of a combinatorial class $(\mathcal{O},s)$, such that 
for all objects $o_1,o_2 \in \mathcal{O}_{\boldsymbol{n}} $ of the same size, the probability to generate $o_1$ is equal
to the probability to generate $o_2$.

The paper is organized as follows.
Section~\ref{sec:def} presents the definition of tree-alphabets, valid words, Lukasiewicz words ordered trees and the links  between the objects.
Section~\ref{sec:alg1} presents a re-description of an algorithm by Alonso et al.~\cite{ARS97}, using the notion of Lukasiewicz words.
Our approach is to prove the algorithm step by step, using simple arguments. 
Section~\ref{sec:dich} present the dichotomic sampling method, which directly generates random valid words, using a linear
number of random bits. 
The last part of the paper follows a simulate-guess-and-prove scheme. We first show some examples of random trees obtained from the generator.
Then, we experimentally and theoretically study the evolution of the tree's height according to the proportion of unary nodes.

\section{Words and Trees \label{sec:def}} 

\subsection{Valid words and Lukasiewicz Words \label{sec:def1}}
This section is devoted to recall the one-to-one map between trees and Lukasiewicz words. This bijection is the central point for the sampling part of the paper.
%\subsection{Words}
Let us recall basic definitions on words. 
An {\em alphabet} $\Sigma$ is a finite tuple $(a_1,...,a_d)$ of distinct symbols called {\em letters}.
A word $w$ defined on $\Sigma$ is a sequence of letters from $\Sigma$.
In the following, $w_i$ denotes the $i$-th letters of the word $w$, $|w|$ its length and for all letter $a \in \Sigma$, $|w|_a$ counts
the occurrences of the  letter $a$ in $w$. 
A \emph{language} defined on $\Sigma$ is a set of words defined of $\Sigma$.

The following new notion of {\em tree-alphabet} will make sense in the next sections. It will allow us to define subclasses of Lukasiewicz words which are in relation to natural combinatorial classes of trees.
\begin{definition}
A {\em tree-alphabet} $\Sigma_f$ is a couple $(\Sigma,f)$ constituted by an alphabet $\Sigma=(a_1,\ldots ,a_k)$ and a function $f: \Sigma \rightarrow \mathbb{N} \cup \{-1\}$ that associates each symbol of $\Sigma$ to an integer such that:
\begin{enumerate}[label=\roman*.]
\item
$f(a_1)=-1$,
\item
$f(a_i)\le f(a_{i+1})$, for $1 \le i < k$.
\end{enumerate}
\end{definition}
%\begin{example} \label{ex:tree-al}
%$\Sigma=\{a,b,c\}$ and
%\begin{itemize}
%\item
%$f(a)=-1$,
%\item
%$f(b)=0 , f(c)= 1$,
%\end{itemize}
%\end{example}
  
We finish this section by introducing Lukasiewicz words. 
%Let us that that purpose, define some intermediate notions. The notion of {\em tree-alphabet} will make sense in the following sections. It will allow us to define subclasses of Lukasiewicz words will be in relation to natural combinatorial classes of trees. A {\em tree-alphabet} $\Sigma_f$ is a couple $(\Sigma,f)$ where $\Sigma=\{a_1,\ldots ,a_k\}$ is 
%an alphabet of $k$ letters and $f: \Sigma \rightarrow \mathbb{N} \cup \{-1\}$
%is a function that associates each symbol of $\Sigma$ to an integer such that:
%\begin{itemize}
%\item
%$f(a_1)=-1$,
%\item
%$f(a_i)\le f(a_{i+1})$, for $1 \le i < k$,
%\end{itemize}

%Let $\Sigma_f$ be a tree-alphabet of $k$-letters and $\boldsymbol{n} \in \mathbb{N}^k$ 
%The tuple $(\Sigma_f,\boldsymbol{n})$ is {\em valid} if 
%$$\sum_{i =1}^k n_i  f(a) = -1 $$
%In this case, the combinatorial class we are interested in is the set of valid tuples $(\Sigma_f,\boldsymbol{n})$.
%
%
%A word $w$ defined over $\Sigma_f$ is valid if 
%\begin{itemize}
%\item
%for all letter $a_i \in \Sigma$, $n_i=|w|_{a_i} $
%\item
%the tuple $(\Sigma_f,\boldsymbol{n})$ is valid.
%\end{itemize}
%We note $\Sigma_f^{\boldsymbol{n}}$ the set of valid word
%
%
%A word $w$ is a {\em Lukasiewicz word} if 
%\begin{enumerate}
%\item
%$w$ is a valid word.
%\item
%for all $i < |w|$, we have 
%$$\sum_{j=1}^i f(w_j) \ge 0 $$
%\end{enumerate}
%For this definition, it is obvious that the Lukasiewicz words are a subset of valid words.
\begin{definition}
A word $w$ on the tree-alphabet  $\Sigma_f=((a_0,...,a_k),f)$ is a {\em $f$-Lukasiewicz word} if :

\begin{enumerate}[label=\roman*.]
%\item
%$f(a_0)=-1$,
%\item
%$f(a_i)\le f(a_{i+1})$, for $1 \le i < k$,
\item
for all $i < k$, we have 
$\sum_{j=0}^i |w|_{a_j}f(a_j) \ge 0 $
\item $\sum_{i =1}^k |w|_{a_i}  f(a_i) = -1 $
\end{enumerate}
\end{definition}

When the condition ii. is verified, we say that the word $w$ if \emph{$f$-valid}. By extension and convenience, we also say that a $k$-tuple $(n_1,\ldots,n_k)$ is $f$-valid $\sum_{i =1}^k f(n_i) = -1 $.

The \emph{Lukasiewicz words $\mathcal{L}_f$} are just the union over all tree-alphabet $\Sigma_f$ of the $f$-Lukasiewicz words.

A classical and useful representation of words on a tree-alphabet is to plot a path describing the evolution of $\sum_{j=1}^i f(w_j)$. Then, a word of size $n$ is valid if and only if the path terminates at position $(n,-1)$ and it is a Lukasiewicz word if and only if the only step that goes under the $x-axis$ is the last one. In particular, these remarks prove that we can verify in linear time if a word is or not a Lukasiewicz word.

 For instance, if $f(a)=-1$, $f(b)=0$ and $f(c)=1$, the following paths represent (from left to right) a Lukasiewicz word, a $f$-valid word and a 
 non valid word:\\
\begin{minipage}{0.28\textwidth}
\begin{tikzpicture}[scale=0.3]
  \draw[very thin,color=gray] (-0.1,-2.1) grid (8,2);
  \draw[->] (-0.2,0) -- (8,0) node[right] {$i$};
  \draw[->] (0,-2) -- (0,2.4) node[above] {{\tiny $\sum_{j=1}^i f(w_i)+1$}};
  \draw[ultra thick] (0,0) -- (1,1) ;
  \draw[ultra thick] (1,1) -- (2,2) ;
  \draw[ultra thick] (2,2) -- (3,2) ;
  \draw[ultra thick] (3,2) -- (4,1) ;
  \draw[ultra thick] (4,1) -- (5,1) ;
  \draw[ultra thick] (5,1) -- (6,0) ;
  \draw[ultra thick] (6,0) -- (7,-1) ;
  \node (w1)  at (0.5,-3.1) {$c$}; 
  \node (w2)  at (1.5,-3.1) {$c$}; 
  \node (w3)  at (2.5,-3) {$b$}; 
  \node (w4)  at (3.5,-3.1) {$a$}; 
  \node (w5)  at (4.5,-3) {$b$}; 
  \node (w6)  at (5.5,-3.1) {$a$};
  \node (w7)  at (6.5,-3.1) {$a$}; 
\end{tikzpicture}
\end{minipage}
\begin{minipage}{0.28\textwidth}
\begin{tikzpicture}[scale=0.3]
  \draw[very thin,color=gray] (-0.1,-2.1) grid (8,2);
  \draw[->] (-0.2,0) -- (8,0) node[right] {};
  \draw[->] (0,-2) -- (0,2.4) node[above] {\textcolor{white}{\tiny $\sum_{j=1}^i f(w_i)+1$}};
  \draw[ultra thick] (0,0) -- (1,0) ;
  \draw[ultra thick] (1,0) -- (2,-1) ;
  \draw[ultra thick] (2,-1) -- (3,-1) ;
  \draw[ultra thick] (3,-1) -- (4,-2) ;
  \draw[ultra thick] (4,-2) -- (5,-1) ;
  \draw[ultra thick] (5,-1) -- (6,0) ;
  \draw[ultra thick] (6,0) -- (7,-1) ;
  \node (w1)  at (0.5,-3) {$b$}; 
  \node (w2)  at (1.5,-3.1) {$a$}; 
  \node (w3)  at (2.5,-3) {$b$}; 
  \node (w4)  at (3.5,-3.1) {$a$}; 
  \node (w5)  at (4.5,-3.1) {$c$}; 
  \node (w6)  at (5.5,-3.1) {$c$};
  \node (w7)  at (6.5,-3.1) {$a$}; 
\end{tikzpicture} 
\end{minipage}
\begin{minipage}{0.28\textwidth}
\begin{tikzpicture}[scale=0.3]
  \draw[very thin,color=gray] (-0.1,-2.1) grid (8,2);
  \draw[->] (-0.2,0) -- (8,0) node[right] {};
  \draw[->] (0,-2) -- (0,2.4) node[above] {\textcolor{white}{\tiny $\sum_{j=1}^i f(w_i)+1$}};
  \draw[ultra thick] (0,0) -- (1,1) ;
  \draw[ultra thick] (1,1) -- (2,0) ;
  \draw[ultra thick] (2,0) -- (3,1) ;
  \draw[ultra thick] (3,1) -- (4,0) ;
  \draw[ultra thick] (4,0) -- (5,0) ;
  \draw[ultra thick] (5,0) -- (6,-1) ;
  \draw[ultra thick] (6,-1) -- (7,-2) ;
  \node (w1)  at (0.5,-3.1) {$c$}; 
  \node (w2)  at (1.5,-3.1) {$a$}; 
  \node (w3)  at (2.5,-3.1) {$c$}; 
  \node (w4)  at (3.5,-3.1) {$a$}; 
  \node (w5)  at (4.5,-3) {$b$}; 
  \node (w6)  at (5.5,-3.1) {$a$};
  \node (w7)  at (6.5,-3.1) {$a$}; 
\end{tikzpicture} 
\end{minipage} 

%\end{center}

Finally, we can give an alternative definition of Lukasiewicz words in the framework of the symbolic method as follows: a word $w$ defined over $\Sigma_f$ is a Lukasiewicz word
if $w=aw_1\ldots w_{f(a)+1}$n where $a \in \Sigma_f$ and $\forall i\le f(a)+1$, $w_i$ is a Lukasiewicz word.
In other word, the combinatorial class of  Lukasiewicz words follow the recursive specification:
$$L=\sum_{a \in \Sigma_f} a L^{f(s)+1}  $$

\subsection{The Tree classes \label{sec:def2}}
%First of all, let us reintroduce trees, in terms of graphs theory, it gives tractable notation and a natural point of view of trees.
%For that purpose, a \emph{directed graph} $G$ is a couple $(V,E)$ such that
%\begin{itemize}
%\item
%$V$ is a set of nodes
%\item
%$E : V\times V$ is a set of edges, that is to say a set of couples of nodes $(u,v)$.
%\end{itemize}
%For a given node $u \in V$, we define $out(u)=\{v \mid (u,v) \in E \}$ and $in(u)=\{v \mid (v,u) \in E \}$.
%The cardinal of $out(u)$ (resp. $in(u)$) is called the {\em out-degree} (resp. {\em in-degree}) of $u$.
%A path is a sequence of edges $a_1\cdots a_n$ such that, $\forall i<n$, we have $a_i=(u,v)$ and $a_{i+1}=(v,w)$.
%A path from $u$ to $v$ is a path $a_1\cdots a_n$ such that $a_1=(u,u')$ and $a_n=(v',v)$.
%
%A {\em tree} is a graph $(V,E)$ such that 
%\begin{itemize}
%\item
%there exists a unique node $r\in V$ whose in-degree is equal to $0$. Such node is called the tree's {\em root}.
%\item
%for all $u\in V\setminus \{r\}$, the in-degree of $u$ is equal to $1$ and there exists a unique path from $r$ to $u$.
%\end{itemize}
%The nodes whose out-degree is equal to $0$ are called leaves.
%
%A tree is called {\em planar} if for all $u \in V$, the set $out(u)$ is totally orderer.

%??? We consider a canonical encoding for planar trees.
%Let $F$ be the set of edges that composes the path starting from a node $u$.
%We encode $F$ with the sequence $s_u$ defined as:
%$$s_u=(u,a_1)s_{a_1}\cdots (u,a_l)s_{a_l} $$
%with $out(u)=\{a_1,\ldots,a_l\}$.
%For a given planar tree $(V,E)$, $E$ is encoded with $s_r$ where $r$ is the tree's root.???

Rooted planar trees are very classical combinatorial objects. Let us recall how we can define them recursively and how this can be described by a formal grammar. Let us begin by the rooted tree class $\mathcal{T}$ over the tree-alphabet $\Sigma_f$ which can be defined as the smallest set verifying:
\begin{itemize}
\item $[x] \in \mathcal{T}$ for every $x\in \Sigma$ such that $f(x)=-1$.
\item Let $x$ such that $f(x)=k$ and $T_1,\cdots,T_k$ in $\mathcal{T}$, then $x[T_1,\cdots,T_k]$ is in $\mathcal{T}$.
\end{itemize} 

%Let $\Sigma_f$ be a tree-alphabet, a \emph{$\Sigma_f$-labelled tree} $(V,E,\lambda)$ is a tree $(V,E)$ with a function $\lambda: V\rightarrow\Sigma $ such that for every $a\in \Sigma$, every node in the preimage $\lambda_d^{-1}(a)$ has out-degree $f(a)+1$.

%Now, we are going to consider the combinatorial class of trees where the size is defined as follows: the size of a tree $T$ having $n_i$ nodes of label $i$ is the vector $(a_0,\cdots,a_d$ where $d$ is the cardinality of the tree-alphabet the maximal out-degree of $T$.

So, the set  $\mathcal{T}$ of all planar \emph{$\Sigma_f$-labelled trees} is a combinatorial class 
whose the size of a tree $T$ is $(|f^{-1}(a_1)|,\cdots,|f^{-1}(a_d)|)$ where $\Sigma=(a_1,...,a_d)$. And just observing the recursive definition, we can specify it from the following symbolic grammar:
$$G= \sum_{s \in \Sigma_f} s G^{f(s)+1} $$

\begin{theorem}[Lukasiewicz]
The combinatorial class of $f$-Lukasiewicz words $\mathcal{L}_f$ is isomorphic to the combinatorial class of trees described by the specification (grammar) $G= \sum_{s \in \Sigma_f} s G^{f(s)+1} $.
\end{theorem}

An explicit bijection can be done as follows: from a $\Sigma_f$-labelled tree $T$, a prefix walk gives a word. This word is a $f$-Lukasiewicz word. Conversely, from a 
$f$-Lukasiewicz word $w$, we build a tree recursively, the root is of degree $f(w_1)+1$ and we continue with the sons as a left-first depth course.

\section{A random sampler as a proof of Tutte's theorem \label{sec:alg1}}

This section is devoted to describe the algorithm that we propose for drawing uniformly a rooted planar tree with a given sequence of degree. The diagram (Fig.\ref{diag}) shows the very simple strategy we adopt. 

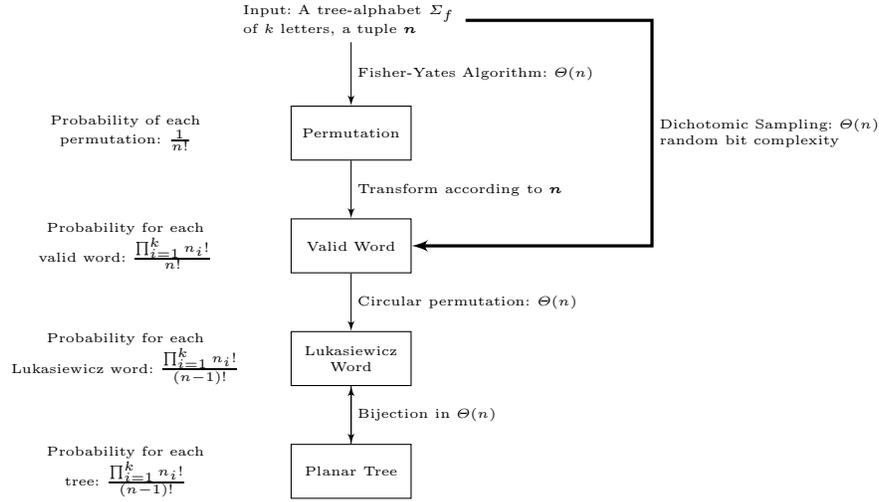
\begin{figure}
{\tiny
\begin{tikzpicture}[auto, node distance=1.5cm,>=latex']

  \node [text width=12em] (start) {Input: A tree-alphabet $\Sigma_f$ of $k$ letters, a tuple $\boldsymbol{n}$};
  
  \node [block, below of=start] (perm) {Permutation};

  %\node [comment, right of=perm, node distance =4cm] (c-perm) {Fisher-Yates Algorithm: $\Theta(n)$};
    \node [comment, left of=perm, node distance =3cm] (p-perm) {Probability of each permutation: $\frac{1}{n!} $};

    \node [block, below of=perm] (vword) {Valid Word};
    \node [comment, left of=vword, node distance =3cm] (p-vword) {Probability for each valid word: $\frac{\prod_{i=1}^k n_i!}{n!} $};

    \node [block, below of=vword] (lword) {Lukasiewicz Word};
    \node [comment, left of=lword, node distance =3cm] (p-lword) {Probability for each Lukasiewicz word: $\frac{\prod_{i=1}^k n_i!}{(n-1)!} $};
 
    \node [block, below of=lword] (tree) {Planar Tree};
    \node [comment, left of=tree, node distance =3cm] (p-tree) {Probability for each tree: $\frac{\prod_{i=1}^k n_i!}{(n-1)!} $};

    \draw [draw,->] (start) --  node[text width=4cm] {Fisher-Yates Algorithm: $\Theta(n)$ } (perm);
    \draw [draw,->, very thick] (start) -- ++(2,0) -- ++(2,0) -- node[text width=3cm] {Dichotomic Sampling: $\Theta(n)$ random bit complexity} ++(0,-3)-- (vword);
    \draw [draw,->] (perm) -- node {Transform according to $\boldsymbol{n}$} (vword);
    \draw [draw,->] (vword) -- node {Circular permutation: $\Theta(n)$} (lword);
    \draw [draw,->] (lword) -- node {Bijection in $\Theta(n)$} (tree);
    \draw [draw,->] (tree) -- node {} (lword);

\end{tikzpicture}
}
\caption{Diagram of the two possible algorithms. The algorithm (Section~\ref{sec:alg1}) using the 
Fisher-Yates algorithm uses $\Theta(n \log n)$ random bits
to generate a random tree with $n$ nodes, but is easy to implement. The algorithm (Section~\ref{sec:dich}) using the Knuth-Yao algorithm~\cite{KY76} or our dichotomic sampling method
use a linear number of random-bit, but doesn't allow us to prove the Tutte's enumerative theorem.}\label{diag}
\end{figure}

The first algorithm contains $4$ steps. The first and the last steps respectively consist in generating a random permutation
using the Fisher-Yates algorithm and the transformation of a Lukasiewicz word into a tree. 
The two other steps are described in the two following subsections.
Each subsection contains an algorithm, the proof of its validity, and its
time and space complexity. We also uses the transformations to obtain enumeration results on each combinatorial object.
Those enumeration results will be useful to prove that the random generator is size-uniform.

\subsubsection{From a permutation to a valid word}
This part is essentially based on the following surjection from permutations to words.
Consider the application $\Phi$ from $\Sigma_n$ the set of permutations of size $n$ to $\mathcal{W}_n$ the words of size $n$ having  for $1\leq i\leq k, n_i$ letters $a_i$ such that :
$$\Phi((\sigma_1,...,\sigma_n))=\phi(\sigma_1)\cdots \phi(\sigma_n)$$ where $\phi(k)=a_i$ if $n_1+\cdots n_{k-1}+1\leq k\leq n_1+\cdots n_{k}$.   
\begin{algorithm}
\KwIn{A tree-alphabet $\Sigma_f$of $k$ letters and a tuple $\boldsymbol{n}$, a permutation $\sigma$ of length $n$}
\KwOut{A tabular $w$ encoding a valid word }

Create a tabular $w$ of size $n$\;
$pos \leftarrow 0$\;
\For{$i \in \{1,\ldots ,k\}$ } {
  \For{$j \in \{1, \ldots, n_i\}$}{
    $w[\sigma_{pos}]\leftarrow a_i$\;
    $pos \leftarrow pos+1$\;
  }
}

\Return{$w $}\;

\caption{From a permutation to a valid word}
\label{algo:vword}
\end{algorithm}

\begin{lemma}
For each valid word $w \in \Sigma_f^{\boldsymbol{n}}$ defined over a $k$ letters alphabet, the number of permutation
associated to $w$ by the Algorithm~\ref{algo:vword} is exactly $\prod_{i=1}^k n_i!$.
\end{lemma}

\begin{proof}
Let us define $m_i =\sum_{j=1}^{i-1} n_i$ and $m_1=0$. The application is invariant by permutation of the values inside $[m_i, \ldots, m_i+n_i]$. So, the cardinality of the kernel is $\prod_{i=1}^k n_i!$.
\end{proof}

%\begin{proof}
%For all integer $i \in \{2, \ldots, k\}$, we define $m_i =\sum_{j=1}^{i-1} n_i$ and $m_1=0$.
%In Algorithm~\ref{algo:vword}, $[m_i, \ldots, m_i+n_i]$ is the interval of values de $pos$ such that $w[\sigma_{pos}] =a_i$.
%Suppose we have two distincts permutations $\sigma$ and $\pi$ on length $n$. Let $w$ (resp. $v$) be the image of $\sigma$ (resp. $\pi$)
%according to Algorithm~\ref{algo:vword}. For any position $j$, we have :
%$$w[j] = v[j] \Longleftrightarrow \exists x,y\  s.t.\ \sigma_x=j,\ \pi_y=j,\exists i \{1,\ldots, k\} x,y \in [m_i,\ldots , m_i+n_i] $$
%This implies that the order of the values in an interval $[m_i,\ldots , m_i+n_i]$ does not change its image with Algorithm~\ref{algo:vword}.
%Therefore, for a given valid word $w$, the number of preimage is equal to $\prod_{i=1}^k n_i!$.
%\end{proof}

\begin{corollary}\label{cor:vword}
The number of valid words in $\Sigma_f^{\boldsymbol{n}}$ is exactly $ \frac{n!}{\prod_{i=1}^k n_i!}$.
\end{corollary}

\begin{lemma}
The time and space complexity of Algorithm~\ref{algo:vword} is $\Theta(n)$.
\end{lemma}

\begin{proof}
The space complexity is linear since we create a tabular of size $n$.
Instructions of line $1,2,5,6$ can be done in constant time. Lines $5$ and $6$ are executed 
$\sum_{i=1}^k n_i$ times, that is to say $n$ times.
\end{proof}

\subsubsection{From a valid word to a Lukasiewicz word}
This part is essentially based on a very simple version of the cyclic lemma which says that among the $n$ circular permutations of a valid word, there is only one which is a  Lukasiewicz word.
%The main idea to this step is that any valid word can be associated to a unique Lukasiewicz word.
%Also, for each Lukasiewicz word $w$ there is exactly $n$ valid words that can be transformed into $w$.
Therefore, if we have a uniform random valid word and transform it into a Lukasiewicz word,
we obtain a uniform Lukasiewicz word.

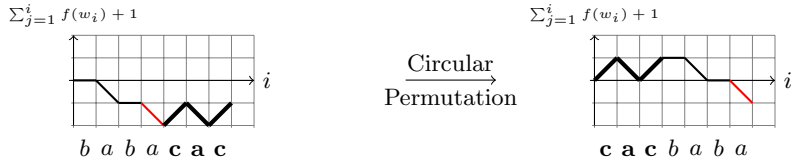
\begin{figure}[h!]
\begin{minipage}{0.4 \textwidth}

\begin{tikzpicture}[scale=0.3]
  \draw[very thin,color=gray] (-0.1,-2.1) grid (8,2);
  \draw[->] (-0.2,0) -- (8,0) node[right] {$i$};
  \draw[->] (0,-2) -- (0,2) node[above] {{\tiny $\sum_{j=1}^i f(w_i)+1$}};
  \draw[thick] (0,0) -- (1,0) ;
  \draw[thick] (1,0) -- (2,-1) ;
  \draw[thick] (2,-1) -- (3,-1) ;
  \draw[thick,color=red] (3,-1) -- (4,-2) ;
  \draw[ultra thick] (4,-2) -- (5,-1) ;
  \draw[ultra thick] (5,-1) -- (6,-2) ;
  \draw[ultra thick] (6,-2) -- (7,-1) ;
  \node (w1)  at (0.5,-3) {$b$}; 
  \node (w2)  at (1.5,-3.1) {$a$}; 
  \node (w3)  at (2.5,-3) {$b$}; 
  \node (w4)  at (3.5,-3.1) {$a$}; 
  \node (w5)  at (4.5,-3.1) {{\bf c}}; 
  \node (w6)  at (5.5,-3.1) {{\bf a}};
  \node (w7)  at (6.5,-3.1) {{\bf c}};

  %\draw[domain=0:3] plot (\x,{square(\x)-1}) node[below right] {$f(x) = x^2$};
\end{tikzpicture}
\end{minipage}
\begin{minipage}{0.15 \textwidth}
\begin{tikzpicture}[scale=0.3]
  \draw[->] (0,0) --  node[above] {Circular} node[below] {Permutation} ++(4,0);
  %\draw[domain=0:3] plot (\x,{square(\x)-1}) node[below right] {$f(x) = x^2$};
\end{tikzpicture}

\end{minipage}
\begin{minipage}{0.4 \textwidth}
\begin{tikzpicture}[scale=0.3]
  \draw[very thin,color=gray] (-0.1,-2.1) grid (8,2);
  \draw[->] (-0.2,0) -- (8,0) node[right] {$i$};
  \draw[->] (0,-2) -- (0,2) node[above] {{\tiny $\sum_{j=1}^i f(w_i)+1$}};
  \draw[ultra thick] (0,0) -- (1,1) ;
  \draw[ultra thick] (1,1) -- (2,0) ;
  \draw[ultra thick] (2,0) -- (3,1) ;
  \draw[thick] (3,1) -- (4,1) ;
  \draw[thick] (4,1) -- (5,0) ;
  \draw[thick] (5,0) -- (6,0) ;
  \draw[thick,color=red] (6,0) -- (7,-1) ;
  \node (w1)  at (0.5,-3.1) {{\bf c}}; 
  \node (w2)  at (1.5,-3.1) {{\bf a}}; 
  \node (w3)  at (2.5,-3.1) {{\bf c}}; 
  \node (w4)  at (3.5,-3) {$b$}; 
  \node (w5)  at (4.5,-3.1) {$a$}; 
  \node (w6)  at (5.5,-3) {$b$};
  \node (w7)  at (6.5,-3.1) {$a$}; 

  %\draw[domain=0:3] plot (\x,{square(\x)-1}) node[below right] {$f(x) = x^2$};
\end{tikzpicture}
\end{minipage}

\caption{An example: the valid word $babacac$ is not a Lukasiewicz word but $cacbaba$ is.
The idea is to find the smallest value of $i$ such that $\sum_{j=1}^i f(w_i)$ is minimal, and compute
the word $w_{i+1}\cdots w_{|w|}w_1\cdots w_i$ \label{fig:path}}
\end{figure}

\begin{lemma}\label{lm:unique}
  For each valid word $w \in \Sigma_f^{\boldsymbol{n}})$, there exists a unique integer $\ell $ such that 
  $w_{\ell+1}\cdots w_n w_1 \cdots w_{\ell }$ is a Lukasiewicz word. 
  Such integer is defined as the smallest integer that minimizes $\sum_{j=1}^{\ell} f(w_j)$.
\end{lemma}
\begin{proof}
Let $w'=w_{\ell+1}\cdots w_n w_1 \cdots w_{\ell }$ be the {\em circular permutation} of $w$ at a position $\ell$.
We notice that $w'$ is a valid word. Let's now picture the path representation of $w$ and $w'$ (see Figure~\ref{fig:path} ).
Let $b(i)$ (resp. ($a(i)$) be the height of the path at position $i$ before (resp. after) the circular permutation. 
In other words:
$$b(i)=\sum_{j=1}^{i} f(w_j)$$ 
$$a(i)=  \begin{cases} b(i)-b(\ell), \text{ for all } i\in \{\ell+1,\ldots ,n\} \\ b(i)-b(\ell)-1, \text{ for all } i\in \{1,\ldots ,\ell\} \end{cases}$$
$w'$ is a Lukasiewicz word iff $a(i)\ge 0$, for all $i\in \{1,\ldots, \ell-1, \ell+1,\ldots ,n\} $, that is to say:
$$a(i) \ge 0 \Longleftrightarrow \begin{cases} b(i)\ge b(\ell), \text{ for all } i\in \{\ell+1,\ldots ,n\} \\ b(i)>b(\ell), \text{ for all } i\in \{1,\ldots ,\ell-1\} \end{cases}$$
This concludes the proof.

\end{proof}

\begin{corollary}
The number of Lukasiewicz words in $\Sigma_f^{\boldsymbol{n}}$ is exactly $ \frac{(n-1)!}{\prod_{i=1}^k n_i!}$.
\end{corollary}

\begin{proof}
From Lemma~\ref{lm:unique} we know that each Lukasiewicz word can be obtained from exactly $n$ valid words.
We conclude using Corollary~\ref{cor:vword}
\end{proof}

\begin{corollary}[Tutte]
The number of trees having $n_i$ of type $i$ and such that $(n_1,...,n_k)$ is $f$-valid is exactly $ \frac{(n-1)!}{\prod_{i=1}^k n_i!}$.
\end{corollary}
\begin{proof}
It is a direct consequence of the bijection between trees and Lukasiewicz words.
\end{proof}

We use the property of Lemma~\ref{lm:unique} to describe an algorithm that transforms any valid word into its associated
Lukasiewicz word.

\begin{algorithm}
\KwIn{A valid word $w$ of length $n$ according to $(\Sigma,f,occ)$}
\KwOut{A tabular $v$ encoding a Lukasiewicz word }

$min \leftarrow cur \leftarrow f(w_1)$\;
$\ell \leftarrow 1$\;

\For{$i \in \{2, \ldots, n\}$}{
  $cur \leftarrow cur + f(w_i)$\;
  \If{$cur < min$}{
    $\ell \leftarrow i$\;
    $min \leftarrow cur$\;
  }
}

Create a tabular $v$ of length $n$\;
\For{$i \in \{1, \ldots, \ell\}$}{
  $v[i+\ell+1]\leftarrow w[i]$\;
}
\For{$i \in \{\ell+1, \ldots, n\}$}{
  $v[i-\ell-1]\leftarrow w[i]$\;
}

\Return{$v $}\;

\caption{From a valid word to a Lukasiewicz word}
\label{algo:lword}
\end{algorithm}

\begin{lemma}
Algorithm~\ref{algo:lword} transforms a valid word into its Lukasiewicz word.
Its time and space complexity is $\Theta(n)$.
\end{lemma}

\begin{proof}
The space complexity is linear since we create a tabular $v$ of size $n$.
The first loop computes unique integer $\ell$ such that $w_{\ell+1}\cdots w_n w_1 \cdots w_{\ell }$ is a Lukasiewicz word, in linear time.
The second and the third loop fill the tabular $v$ of length $n$ such that $v=w_{\ell+1}\cdots w_n w_1 \cdots w_{\ell }$.
\end{proof}

\subsection{First algorithm}\label{LtoT}

\begin{theorem}
Algorithm~\ref{algo:final} is a random planar tree generator. 
Its time and space arithmetic complexity is linear.
%Assume that the number of different types of nodes is $k(n)$, then its time (in terms of numbers of  random bits) complexity is $\Theta(n\ln(k(n)))$. The space complexity is linear.
\end{theorem}

\begin{algorithm}
\KwIn{A tree-alphabet $\Sigma_f$ of $k$ letters and a tuple $\boldsymbol{n}$}
\KwOut{A random planar tree satisfying $\Sigma_f$ and $\boldsymbol{n}$ }

Generate a random permutation $\sigma$ using Fisher-Yates Algorithm\;
Transform $\sigma$ into a valid word $w$\;
%Generate a valid word $w$ using dichotomic sampling algorithm \;
Transform $w$ into a Lukasiewicz word $v$\;
Transform $v$ in a planar tree $t$

\Return {$t$}\;
\caption{Random Planar Tree Generator}
\label{algo:final}
\end{algorithm}

\section{The dichotomic sampling method \label{sec:dich}}

Using the diagram of Figure~\ref{diag} above, we arrive at the algorithm \ref{algo:final}. However, this algorithm is not optimal in the number of random bits because drawing the permutation consumes more bits than necessary. We shall describe another method to generate valid words more efficiently. The problem is just to draw a $f$-valid word from a $f$-valid tuple $\boldsymbol{n}=(n_1,\ldots,n_k)$. For that purpose, consider the random variable $A$ on the letters of $\Sigma$, assume that $A_1$ follows the distribution $D_{\boldsymbol{n}}$: $Prob(A_1=a_i)=\dfrac{n_i}{\sum_i n_i}$, draw $A_1$ (says $A_1=a_j$) and put it in the first place in the word (i.e. $w_1=a_j$). Now, $A_2$ is conditioned by $A_1$, just by decrease by one $n_j$, again draw $A_2$ and put it in the second place, and so on. This algorithm is described below (see Algorithm~\ref{algo:draw1}).  It is clear that the built word is a $f$-valid word, because it contains exactly the good number of each letters. Now, it is drawn uniformly, indeed, in a uniform $f$-valid word, the first letter follows exactly the distribution $D_{\boldsymbol{n}}$, the sequel follows directly by induction.

\begin{algorithm}
\KwIn{A tree-alphabet $\Sigma_f$ of $k$ letters and a tuple $\boldsymbol{n}$}
\KwOut{A tabular $w$ encoding a valid word }

Create a tabular $w$ of size $n$\;
%$s \leftarrow \sum n_i$\;
\For{$i \in \{1,\ldots ,n\}$ } {
  $k\leftarrow Distrib(\boldsymbol{n})$ ($k$ is drawn according to the distribution $D_{\boldsymbol{n}}$)\;
    $w[i] \leftarrow a_k$\;
    $\boldsymbol{n} \leftarrow \boldsymbol{n}-\boldsymbol{e}_k$ ($\boldsymbol{e}_k$ denotes the $k$-th canonical vector)\;
}

\Return{$w $}\;

\caption{From a tuple $\boldsymbol{n}$ to a valid word}
\label{algo:draw}
\end{algorithm}

So, to obtain a random-bit optimal sampler, we just need to have a optimal sampler for general discrete distribution. But, it is exactly the result obtained by Knuth-Yao~\cite{KY76}. Therefore we have the following result:

\begin{theorem}
By replacing the two first steps of Algorithm~\ref{algo:final} by Algorithm~\ref{algo:draw1}, one obtains
a random-bit optimal sampler for rooted planar tree with a given sequence of degree. 
\end{theorem}

Nevertheless, according to the authors, the Knuth-Yao algorithm can be inefficient in practice (because it needs to solve the difficult question to generate infinite DDG-trees). There is a long literature on it which is summarized in the book of L. Devroye~\cite{D86}. Let just mention the interval sampler from \cite{H97} and the alias methods \cite{V91,W77,MTW04}.

We propose in the sequel a nearly optimal and very elementary algorithm, called \emph{dichotomic sampling}, 
to draw a random variable $X$ following a given discrete distribution of $k$ parts, say, $Prob(X=x_i)=p_i$ for $1\leq i\leq k$. 
%% Let us denote $s_j=\sum_{i\leq j}p_i$ and let us express $X$ as $[s_1,\cdots,s_{k-1}]$. 
%% Now, let $m$ be the largest integer such that $s_{m}\leq 1/2$ and $m=0$ if there is no such a number. 
%% Then to draw $X_1=[s_1,\ldots,s_{k-1}]$, we can flip a coin (Bernoulli choice of parameter 1/2), and if the result is heads then 
%% if $m=0$ return $x_0$ or draw in $X_2=[2s_1,\ldots,2s_{m},1]$, if the result is tails and if $m=k-1$ then return $x_k$ else draw in 
%% $X_2=[2s_{m+1}-1,\ldots,2s_{k-1}-1]$ and continue recursively while $X_\mu\neq\emptyset$
% then assume that $X_{\mu-1}=[f(s_p),\ldots,f(s_q)]$ (for some affine function $f$) and return $x_p$ if we have chosen the left and $x_q$ otherwise.

\begin{algorithm}
\KwIn{a tuple $\boldsymbol{n}=(n_1,\ldots,n_k)$ such that $n=\sum_{i=1}^k n_i$}
\KwOut{An integer between $1$ and $k$}
$i\leftarrow 1$\;
$j\leftarrow k$\;
$min \leftarrow 0$\;
$max \leftarrow n$\;
\While{$i \neq j$}{
\If{DrawRandomBit is equal to $1$}{
  $tmp \leftarrow min$\;
  $min  \leftarrow \frac{min+max}{2}$\;
  \While{$min>(tmp+n_i)$}{
    $i\leftarrow i+1$\;
    $tmp \leftarrow tmp+1$\;
  }
}
\Else{
 $tmp \leftarrow max$\;
  $max  \leftarrow \frac{min+max}{2}$\;
  \While{$max<(tmp-n_i)$}{
    $j\leftarrow j-1$\;
    $tmp \leftarrow tmp-n_j$\;
  }

}
}
\Return{$i $}\;

\caption{Dichotomic sampling}
\label{algo:draw1}
\end{algorithm}

The dichotomic sampling algorithm implies the following induction for $C_n$ the mean number of flip needed for drawing when there are $n+1$ parts~: $C_1=2$ and $C_k=1+\frac{1}{2}\max_{0\leq k \leq m}(C_m+C_{k-m})$. First, let us assume that $C_k$ is concave, so let us consider $\tilde{C}_k=1+\frac{1}{2}({\tilde{C}_{\lfloor\frac{k}{2}\rfloor}+\tilde{C}_{\lceil\frac{k}{2}\rceil}})$. A short calculation shows that $\tilde{C}_{n}=\lfloor\ln_2(n-1)\rfloor+1+\dfrac{n}{2^{\lfloor\ln_2(n-1)\rfloor}}$. Now, by induction, we can easy check that $C_k=\tilde{C}_k$.  So, in particular, $C_k\leq 2+\ln_2(k)$.

\begin{figure}[htb]
\begin{center}
\includegraphics[scale=0.6]{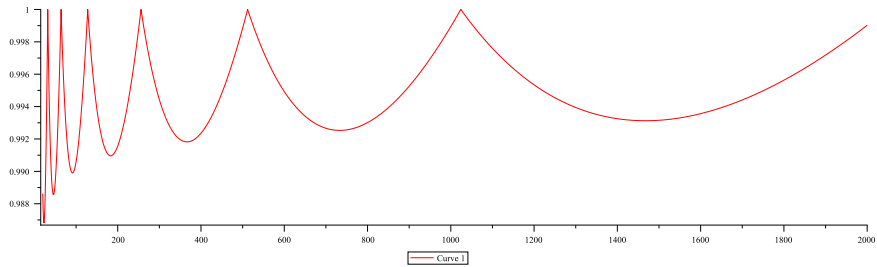} 
\caption{Graphic for Mean Cost $\dfrac{C_k}{2+\ln_2(k)}$}
\label{fig:cout}
\end{center}
\end{figure}

 Note that the sequence $C_k$ can also be analyzed by classical Mellin transform techniques and the periodic phenomena we show in the figure \ref{fig:cout} is quite familiar.

\section{Simulate-Guess-and-prove : Analysis of height}

In this section, we study experimentally and theoretically the height of random Motzkin trees (unary-binary) when the proportion of unary nodes fluctuates.

Figure~\ref{fig:rte} shows example of random Motzkin trees  generated with the algorithm from Section~\ref{sec:alg1}, with different
proportions of unary nodes. Figure~\ref{fig:curv-height} shows the evolution of the height of trees when one increases the proportions of unary nodes.

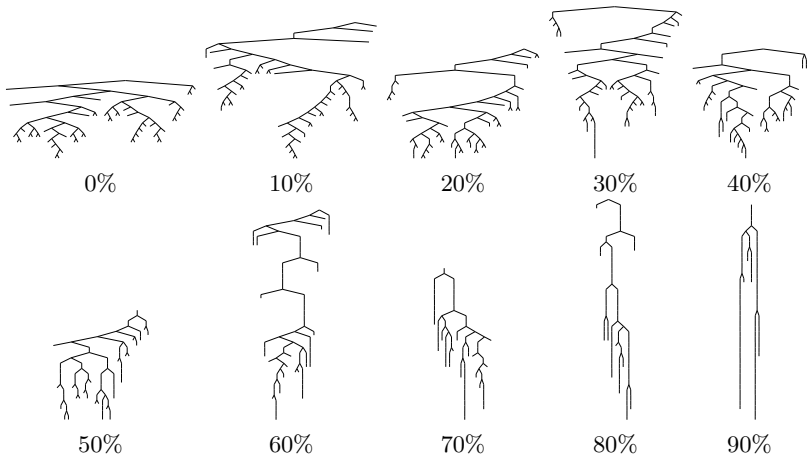
\begin{figure}
\begin{tabular}{c c c c c}
\begin{tikzpicture}[sibling distance=1pt,level distance=10pt,scale=0.2] \Tree [ [ {} [ [ [ {} [ [ [ [ [ [ {} [ {} {} ]]{} ][ [ {} {} ][ {} {} ]]][ [ {} [ [ {} [ {} [ [ {} {} ]{} ]]]{} ]][ [ {} [ {} {} ]][ {} {} ]]]]{} ][ {} {} ]]]{} ][ [ [ [ {} [ {} {} ]][ {} {} ]][ {} [ {} [ [ [ [ [ {} {} ]{} ]{} ]{} ][ {} [ [ {} [ {} {} ]]{} ]]]]]][ {} [ [ [ {} {} ]{} ]{} ]]]]][ [ {} {} ]{} ]]\end{tikzpicture} &
\begin{tikzpicture}[sibling distance=1pt,level distance=10pt,scale=0.2] \Tree [ [ [ [ [ [ [ {} ]][ {} [ [ [ {} [ [ [ [ {} [ [ [ [ {} ]{} ][ {} {} ]]{} ]][ {} {} ]]{} ]{} ]]][ [ [ {} {} ]{} ][ [ {} {} ][ {} [ [ [ [ [ [ [ [ [ [ {} [ [ {} [ [ {} [ {} [ {} {} ]]]{} ]][ {} ]]][ {} {} ]]{} ][ {} {} ]]{} ]{} ]{} ]][ {} [ {} [ [ [ {} [ {} [ {} {} ]]]]]]]][ [ {} {} ]]]]]]]]]{} ]]{} ]{} ]\end{tikzpicture} &
\begin{tikzpicture}[sibling distance=1pt,level distance=10pt,scale=0.2] \Tree [ [ [ [ [ [ [ [ [ [ {} {} ]{} ]]{} ]][ [ [ [ [ [ [ [ {} [ {} [ [ [ [ {} [ [ [ {} ]]{} ]][ [ [ [ {} {} ]{} ]{} ]]][ {} [ {} {} ]]]{} ]]][ [ [ [ [ [ [ {} ][ [ [ {} ]][ {} {} ]]]][ [ {} [ {} ]][ [ [ {} {} ]][ {} {} ]]]]][ {} ]][ {} [ {} {} ]]]]{} ][ [ {} {} ]]][ [ [ {} {} ]]]]{} ]]]]][ {} ]]{} ][ {} {} ]]\end{tikzpicture} &
\begin{tikzpicture}[sibling distance=1pt,level distance=10pt,scale=0.2] \Tree [ [ [ {} [ {} [ [ {} ][ {} ]]]]][ [ [ [ [ [ [ [ {} [ [ [ [ {} [ {} [ [ [ [ {} [ [ {} ][ [ [ {} [ [ {} ]]][ {} [ [ [ [ [ [ [ {} ]]]]]]]]]]]]{} ]{} ][ {} ]]]]][ [ [ [ [ {} {} ][ {} [ [ {} ][ [ [ [ [ [ {} ][ {} ]]]{} ]{} ][ {} {} ]]]]][ [ [ {} [ [ {} [ [ {} ]]]]]]]]{} ]]]{} ]]][ {} ]]{} ]]{} ][ {} ]]{} ]]\end{tikzpicture} &
\begin{tikzpicture}[sibling distance=1pt,level distance=10pt,scale=0.2] \Tree [ [ [ [ {} [ [ [ {} {} ][ [ [ [ [ [ {} ]{} ]{} ]{} ]][ [ [ [ [ [ [ [ [ [ {} ]]][ [ {} ]]]]]][ [ [ [ {} ]][ [ [ [ [ [ {} ]][ [ {} ][ [ {} ][ [ {} ][ [ {} ]]]]]][ {} ]]]{} ]]]]]{} ]]][ [ [ [ [ [ [ [ [ [ {} ][ [ {} {} ]]]][ {} [ {} {} ]]]]][ [ [ [ [ {} ]]{} ][ [ {} ]]]{} ]]{} ]]]]]]][ [ [ {} ]][ [ {} ]]]]\end{tikzpicture} \\
0\% & 10\% & 20\% & 30\% & 40\% \\
\begin{tikzpicture}[sibling distance=1pt,level distance=10pt,scale=0.2] \Tree [ [ [ [ [ [ [ {} [ [ [ [ [ [ [ [ {} [ [ [ [ {} ][ [ [ {} ][ {} ]]]]]]]]]]][ [ [ [ [ {} [ [ {} {} ]]]]][ [ [ [ {} [ {} ]]]{} ]]]]][ [ [ [ [ [ [ [ [ {} {} ]][ [ [ [ [ {} ]]]][ [ {} [ [ {} ]]]]]]]]][ [ [ [ [ {} ]]]]]]]]]]][ [ [ [ {} [ [ [ [ {} ]]]]]]][ [ {} {} ]]]][ {} [ {} ]]][ [ {} ]]]][ [ {} [ {} ]]]]]\end{tikzpicture} &
\begin{tikzpicture}[sibling distance=1pt,level distance=10pt,scale=0.2] \Tree [ [ [ [ [ [ [ {} ]]][ [ [ {} ]][ [ [ [ [ [ [ [ [ [ [ [ {} ][ [ [ [ [ [ [ [ [ [ [ [ {} ]]][ [ [ [ {} [ [ [ [ [ [ [ {} [ [ [ [ {} ]]]]]]][ [ {} ]]]][ {} ]]{} ]]{} ]][ [ [ [ [ {} {} ][ [ [ [ [ [ {} ]]]]]]]]][ [ [ {} ]]]]]][ [ [ [ [ [ {} ]]]]]]][ {} ]]]]]]]]]]]]]][ [ {} ]]]]]]]]][ {} ]]{} ][ [ [ [ {} ]]]]]\end{tikzpicture} &
\begin{tikzpicture}[sibling distance=1pt,level distance=10pt,scale=0.2] \Tree [ [ [ [ [ [ [ [ [ [ [ {} ]]]]]]]]][ [ [ [ [ [ [ [ [ [ [ [ [ [ {} ]]]]]][ [ [ [ {} ]][ [ [ [ [ [ [ {} ]]]]]][ [ [ [ [ [ {} ]]]{} ]]]]]]][ [ [ [ [ {} ]][ [ [ [ [ [ [ [ [ {} ]]][ [ [ [ [ [ [ [ [ [ [ {} ]]]]]]]]]]]]]][ [ {} ]]]][ [ [ [ {} [ [ [ [ {} ][ [ [ {} ][ [ [ [ [ {} ]]]]]]]]{} ]]]]]{} ]]]]]]]]]]]]]]\end{tikzpicture} &
\begin{tikzpicture}[sibling distance=1pt,level distance=10pt,scale=0.2] \Tree [ [ {} ][ [ [ [ [ [ [ [ [ [ {} ]][ [ [ [ [ [ [ [ [ [ [ [ [ [ [ [ [ [ {} ]]]][ [ [ [ {} ]]]]]]]]]][ [ [ [ [ [ [ [ [ [ [ [ [ [ [ [ [ {} ]]]]]]]]]][ [ [ [ {} [ [ [ [ [ [ [ [ [ {} ]]]]]]][ [ {} ]]]]]]][ [ [ [ [ [ [ [ [ [ [ [ [ [ [ [ [ {} ]]]]]][ [ [ [ {} ]]]]]]]]]]]]]]]]]]]]]]]]]]]]]]]]][ [ [ {} ]]]]]]]]]]\end{tikzpicture} &
\begin{tikzpicture}[sibling distance=1pt,level distance=10pt,scale=0.2] \Tree [ [ [ [ [ [ [ [ [ [ [ [ [ [ [ [ [ [ [ [ [ [ [ [ [ [ [ [ [ [ [ [ [ [ [ [ [ [ [ {} ]]]]]]]]]]]]]]]]]]]]]]]]][ {} ]]]]]]]]][ [ [ [ [ {} ]][ [ [ [ [ [ {} ]]]]]]]]]][ [ [ [ [ [ [ [ [ [ [ [ [ [ [ [ [ [ [ [ [ [ [ [ [ [ [ [ [ [ [ [ [ [ [ [ {} ]]]]]]]]]]]]]]]]]]]][ [ [ [ [ [ [ [ {} ]]]]]]]]]]]]]]]]]]]]]]]]]]]]]\end{tikzpicture} \\
50\% &60\% & 70\% & 80\% & 90\% \\
\end{tabular}
\caption{Example of Motzkin trees with 101 nodes generated with our algorithm, where the proportion of unary nodes varies from 0\% to 90\%. \label{fig:rte}}
\end{figure}

\begin{figure}
\begin{center}
\rotatebox{-90}{\includegraphics[scale=0.3]{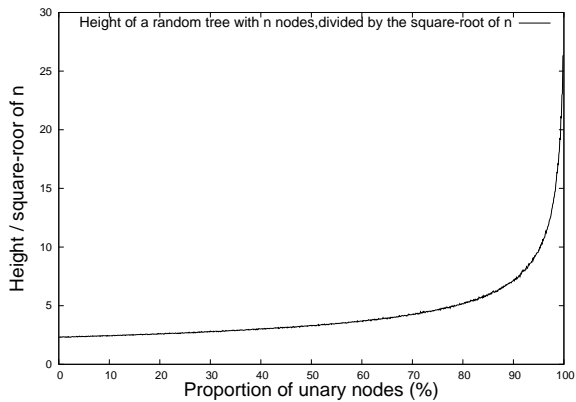}}
\caption{In this example, all random trees have $n=1000$ nodes. For each proportion of unary nodes, varying from $0$ to $99,9$ percent, 
$10\, 000$ Motzkin trees have been generated. The curve shows how the average height of Motzkin trees, divided by the square-root of $n$. 
\label{fig:curv-height} }
\end{center}
\end{figure}

In the following, we study the height of Motzkin trees according to the proportion of unary nodes, using
exclusively probabilistic arguments.

The continuum random tree (CRT) is a random continuous tree 
defined by Aldous \cite{A93}, which is closely related to Brownian 
motion. In particular, the height of the CRT has the same law as the maximum 
of a Brownian excursion. The CRT can be viewed as the renormalized 
limit of several models of large trees, in particular, critical 
Galton-Watson trees with finite variance conditioned to have a large 
population \cite{G98,D03,M08}. Our model does not exactly fit into 
this framework, however, it is quite clear that the proofs can be adapted 
to our situation. We show here a convergence result related to the height 
of Motzkin trees.

\begin{thm}
Let $(c_n, n\geq 1)$ be a sequence of integers such that
$c_n=o(n)$ and $(\log n)^2=o(c_n)$. Then one can construct, on a single 
probability space, a family $(T_n, n\geq 1)$ of random trees 
and a random variable $H>0$ such that 

(i) for every $n\geq 1$, $T_n$ is a uniform Motzkin tree with $n$ 
vertices and $c_n+1$ leaves. 

(ii) $H$ has the law of the height of the CRT 

(iii) almost surely, 
$$
\frac{\sqrt{c_n}}{n}height(T_n)\to H
$$
\end{thm}

\begin{proof}
The proof's idea is the following:
\begin{itemize}
\item
A Motzkin tree can been seen as a binary tree with $2c_n+1$ nodes in which we each node can be replaced by a sequence
of unary nodes. If $n$ is the size of the Motzkin tree, then the number of unary node is $n-2c_n-1$.
\item
The height of a leaf in the Motzkin tree is equal its length in the binary tree plus the lengths of the sequences of unary nodes
between the leaf and the tree's root.
\item
We study the probability that the lengths sum of the sequences of unary nodes between a given leaf and the tree's root is equal to a given value.
\item
We use this result to frame the generic height of $T_n$. 
\end{itemize}

We assume that $(c_n, n\geq 1)$ is non-decreasing, otherwise, the
proof can be easily adapted.
Let us call the skeleton of a  Motzkin tree the binary tree obtained by 
forgetting the vertices having one child.
Denote by $S_n$ the skeleton of $T_n$. For a leaf $l$, let $d(l)$ be
the distance of $l$ to the root in $S_n$ and  $D(l)$
the distance of $l$ to the root in $T_n$.

First,  one can construct the sequence $(S_n, n\geq 1)$ by R\'emy's algorithm
\cite{R85} and it can be shown that $S_n$ converges in a strong sense 
to a CRT \cite{CH14}, in particular,
$$
\frac{height(S_n)}{\sqrt{c_n}}\to H
$$
where  $H$ has the law of the height of the CRT.

Next, for every $n\geq 1$, one can obtain $T_n$ from $S_n$ by replacing 
each edge $e$
of $S_n$ with a ``pipe'' containing $X_e$ nodes of degree 2. The family 
$(X_e)$ is a $2c_n$-dimensional random vector with non-negative integer 
entries, and it is uniformly distributed over all vectors of this kind 
such that the sum of the entries is $n-2c_n-1$. Let us denote 
$(X_e)=(X_1,\ldots, X_{2c_n})$ (we should write 
$(X^{(n)}_1,\ldots, X^{(n)}_{2c_n})$ but we want to make the notation lighter).

It is a classical remark that the random variable $(X_1,\ldots, X_{2c_n})$ has the same law as 
$(Y_1,\ldots, Y_{2c_n})$ conditional on the event $\sum_i Y_i=n-2c_n-1$,
where the $Y_i$ are independent, geometric random variables with mean
$$
m_n=\frac{n-2c_n-1}{2c_n}
$$ 
Moreover, since the sum $\sum_i Y_i$ has mean $n-2c_n-1$ and variance 
$\sim c_n m_n^2$, a classical local limit theorem \cite{G48}
tells us that there exists 
a constant $c>0$ such that for every $n\geq 1$,
\begin{equation}\label{cond}
\bbp(\sum_i Y_i=n-2c_n-1)\geq \frac{1}{c\sqrt{c_n}m_n}
\end{equation}

Fix $\varepsilon>0$. Pick at random a realization of R\'emy's algorithm,
yielding a sequence of binary trees $(S_n, n\geq 1) $ such
for every $n\geq 1$, $S_n,$ has $c_n+1$ leaves. 
Then almost surely, there exists $H>0$ such that the height of $S_n$, which
we denote $h_n$, satisfies
\begin{equation}\label{cvheight}
\frac{h_n}{\sqrt{c_n}}\to H
\end{equation}
From now on, since we have chosen our sequence $(S_n, n\geq 1) $, 
the symbols $\bbp$ and $\bbe$ will refer to the probability 
and expectation with respect to the random variables
$(X_i)$, $(Y_i)$, $(Z_i)$. 

If a leaf $l$ in $S_n$  is at a distance $d(l)$ from the root, then
its  distance $D(l)$ from the root in $T_n$ is the sum of $d(l)$ random 
variables in the family $(X_e)$. Therefore,
\begin{eqnarray*}
\bbp(D(l)=k)&=&\bbp(X_1+\ldots+X_{d(l)}=k)\\
&=&\bbp(Y_1+\ldots+Y_{d(l)}=k|\sum_i Y_i=n-2c_n-1)\\
&=&\frac{\bbp(Y_1+\ldots+Y_{d(l)}=k,\sum_i Y_i=n-2c_n-1)}
{\bbp(\sum_i Y_i=n-2c_n-1)}\\
&\leq& 
\frac{\bbp(Y_1+\ldots+Y_{d(l)}=k)}
{\bbp(\sum_i Y_i=n-2c_n-1)}
\end{eqnarray*}
The right-hand side is maximized when  $d(l)=h_n$. 
We shall now use  independent, exponential random variables 
$(Z_1,\ldots, Z_{2c_n})$ with mean
\begin{equation}\label{mu}
\mu_n=\frac{1}{\log(m_n/(m_n-1))}
\end{equation}
It is easy to check that for every integer $k\geq 0$, 
$$\bbp(Z_1\in[k,k+1])=\bbp(Y_1=k).$$ Therefore, we can define $Y_i$ as the
integer part of $Z_i$ for each $i$.
Since $Z_i\geq Y_i$ for each $i$,
$$
 \bbp\left(\frac{Y_1+\ldots+Y_{h_n}}{\sqrt{c_n}m_n}\geq (1+\varepsilon)H\right)
\leq \bbp\left(\frac{Z_1+\ldots+Z_{h_n}}{\sqrt{c_n}m_n}\geq (1+\varepsilon)
H\right)
$$
Substracting the expectation,
\begin{eqnarray*}
&&\bbp\left(\frac{Z_1+\ldots+Z_{h_n}}{\sqrt{c_n}m_n}\geq (1+\varepsilon)H\right)\\
&&=
\bbp\left(\frac{Z_1+\ldots+Z_{h_n}-h_n\mu_n}{\sqrt{c_n}m_n}
\geq (1+\varepsilon)H-\frac{h_n\mu_n}{\sqrt{c_n}m_n}\right)
\end{eqnarray*}
Because of \eqref{mu} and \eqref{cvheight}, we have, for $n$ large enough,
$$
\left(1-\frac{\varepsilon}{2}\right)H\leq 
\frac{h_n\mu_n}{\sqrt{c_n}m_n}\leq \left(1+\frac{\varepsilon}{2}\right)H
$$
This entails that for $n$ large enough,
$$
(1+\varepsilon)H-\frac{h_n\mu_n}{\sqrt{c_n}m_n}
\leq \frac{\varepsilon H}{2}
$$
and therefore, 
\begin{eqnarray*}
&&
\bbp\left(\frac{Z_1+\ldots+Z_{h_n}-h_n\mu_n}{\sqrt{c_n}m_n}
\geq (1+\varepsilon)H-h_n\mu_n\right)\\
&&\leq 
\bbp\left(\frac{Z_1+\ldots+Z_{h_n}-h_n\mu_n}{\sqrt{c_n}m_n}
\geq \frac{\varepsilon H}{2}\right)
\end{eqnarray*}
We now use the Laplace transform: for every $\lambda>0$,
$$
\bbe \exp(\lambda Z_1-\mu_n)=\frac{e^{-\lambda \mu_n}}{1-\lambda\mu_n}
$$
The Markov inequality yields
$$
\bbp\left(\frac{Z_1+\ldots+Z_{h_n}-h_n\mu_n}{\sqrt{c_n}m_n}
\geq \frac{\varepsilon H}{2}\right)\leq 
\left(\frac{e^{-\lambda \mu_n}}{1-\lambda\mu_n}\right)^{h_n}
\exp\left(-\lambda\sqrt{c_n}m_n \frac{\varepsilon H}{2}\right)
$$
Let $(t_n)$ be a sequence of positive real numbers such that $t_n$ tends to 0 and 
that $\sqrt{c_n} t_n/\log n$ tends to infinity.
Choose $\lambda$ such that $\lambda \mu_n=t_n$. Then, 
$$
\bbp\left(\frac{Z_1+\ldots+Z_{h_n}-h_n\mu_n}{\sqrt{c_n}m_n}
\geq \frac{\varepsilon H}{2}\right)\leq 
\left(\frac{e^{-t_n}}{1-t_n}\right)^{h_n}
\exp\left(-\frac{t_n\sqrt{c_n}m_n\varepsilon H}{2\mu_n} \right)
$$
For $n$ large enough, we have $m_n\geq \mu_n/2$ and
$$
\frac{e^{-t_n}}{1-t_n}\leq 1+2t_n^2
$$
Therefore, for $n$ large enough
$$
\bbp\left(\frac{Z_1+\ldots+Z_{h_n}-h_n\mu_n}{\sqrt{c_n}m_n}
\geq \frac{\varepsilon H}{2}\right)\leq 
(1+2t_n^2)^{h_n}
\exp\left(-\frac{\varepsilon Ht_n\sqrt{c_n}}{4}\right)
$$
Summing up, if $n$ is large enough, then for every leaf $l$, 
$$
\bbp\left(\frac{D(l)}{\sqrt{c_n}m_n}\geq (1+\varepsilon)H\right)
\leq 
\frac{(1+2t_n^2)^{h_n}\exp\left(-\frac{\varepsilon Ht_n\sqrt{c_n}}{4}\right)}
{\bbp(\sum_i Y_i=n-2c_n-1)}
$$
Using the estimate \eqref{cond},
$$
\bbp\left(\frac{D(l)}{\sqrt{c_n}m_n}\geq (1+\varepsilon)H\right)
\leq 
c\sqrt{c_n}m_n
(1+2t_n^2)^{h_n}{\exp\left(-\frac{\varepsilon Ht_n\sqrt{c_n}}{4}\right)}
$$
Since there are $c_n+1$ leaves, and since the probability of the union 
is less that the sum of the probabilities, for $n$ large enough,
$$
\bbp\left(\frac{height(T_n)}{\sqrt{c_n}m_n}\geq (1+\varepsilon)H\right)
\leq 
 c(c_n+1)\sqrt{c_n}m_n
(1+2t_n^2)^{h_n}\exp\left(-\frac{\varepsilon Ht_n\sqrt{c_n}}{4}\right)
$$
The upper bound can be rewritten as
$$
c
\exp\left(h_n\log(1+2t_n^2)-\frac{\varepsilon Ht_n\sqrt{c_n}}{4}+\log m_n
+\frac{3}{2}\log (c_n+1)\right)
$$
Recall that for $n$ large enough, 
$$
h_n\leq (1+\varepsilon/2)H\sqrt{c_n}
$$
and then our bound becomes
$$
\exp\left(H\sqrt{c_n}\left[(1+\varepsilon/2)\log(1+2t_n^2)-
\frac{\varepsilon t_n}{4}\right]+\log m_n
+\frac{3}{2}\log (c_n+1)\right)
$$
Since $t_n\to 0$, for $n$ large enough,
$$
[(1+\varepsilon/2)\log(1+2t_n^2)-
\frac{\varepsilon t_n}{4}]\geq -\frac{\varepsilon t_n}{8}
$$
and so for $n$ large enough, our bound becomes
$$
b_n=\exp\left(\frac{-H\varepsilon t_n\sqrt{c_n}}{8}+\log m_n
+\frac{3}{2}\log (c_n+1)\right)
$$
Now because of the assumption
that $\sqrt{c_n} t_n/\log n\to \infty$, we remark that $\sum b_n<\infty$. Thus by the Borel-Cantelli
lemma, almost surely, conditional on the sequence $ (S_n)$, for $n$ 
large enough,
$$
\frac{height(T_n)}{\sqrt{c_n}m_n}\leq (1+\varepsilon)H
$$
Integrating with respect to the law of the sequence $ (S_n)$, we find that
almost surely, there exists a random variable $H$ which has the
law of the height of the CRT and such that  for $n$ large enough,
$$
\frac{height(T_n)}{\sqrt{c_n}m_n}\leq (1+\varepsilon)H
$$
Likewise, one shows that almost surely, for $n$ large enough,
$$
\frac{height(T_n)}{\sqrt{c_n}m_n}\geq (1-\varepsilon)H
$$
This being true for every positive $\varepsilon$, our result is established.
\end{proof}

{\bf Remark}
In the case when the number of leaves is proportional to the number of 
vertices, $c_n\sim k n$ for some constant $k\in (0, 1/2]$, 
it can be shown by the same arguments that 
$\frac{height(T_n)}{\sqrt{n}}$ converges to $2(1-k)H$.

In the case when $(\log n)^2/c_n$ does not tend to 0, a refinement in the
proof is necessary. Typically, replacing
the inequality \eqref{cond} with a stochastic
domination argument would prove that the height of the tree converges 
in distribution whenever $c_n\to \infty$. To prove an almost sure convergence, 
a more detailed construction would be needed.\\

{\bf General case}

We only assume that $c_n$ tends to infinity. The construction of the skeleton
and the convergence of R\'emy's algorithm still hold. 
The representation of the variables $X_i$ as conditioned versions of the $Y_i$
can be refined in the following manner:
\begin{eqnarray*}
&&\bbp(X_1+\ldots+X_{d(l)}\geq A)\\
&&=\bbp(Y_1+\ldots+Y_{d(l)}\geq A|\sum_i Y_i=n-2c_n-1)\\
&&=\sum_{k=A}^\infty \bbp(Y_1+\ldots+Y_{d(l)}=k|\sum_i Y_i=n-2c_n-1)\\
&&=\sum_{k=A}^\infty\bbp(Y_1+\ldots+Y_{d(l)}=k|\sum_{i=d(l)}^{2c_n} Y_i=n-2c_n-1-k)\\
&&=\sum_{k=A}^\infty
\frac{\bbp(Y_1+\ldots+Y_{d(l)}=k,\sum_{i=d(l)}^{2c_n} Y_i=n-2c_n-1-k)}
{\bbp(\sum_i Y_i=n-2c_n-1)}\\
&&=\sum_{k=A}^\infty
\frac{\bbp(Y_1+\ldots+Y_{d(l)}=k)\bbp(\sum_{i=d(l)}^{2c_n} Y_i=n-2c_n-1-k)}
{\bbp(\sum_i Y_i=n-2c_n-1)}\\
\end{eqnarray*}

Gnedenko's result also gives the existence of a real $C$ such that for 
every integer $k$,
\begin{equation}\label{cond'}
\bbp(\sum_{i=d(l)}^{2c_n} Y_i=n-2c_n-1-k))\leq \frac{C}{\sqrt{c_n-d(l)}m_n}
\end{equation}
From \eqref{cond} and \eqref{cond'} we deduce that if $d(l)\leq c_n/2$,
the following stochastic domination bound hols:
\begin{eqnarray*}
&&\bbp(Y_1+\ldots+Y_{d(l)}\geq A|\sum_i Y_i=n-2c_n-1)\\
&&=\sum_{k=A}^\infty
\frac{\bbp(Y_1+\ldots+Y_{d(l)}=k)\bbp(\sum_{i=d(l)}^{2c_n} Y_i=n-2c_n-1-k)}
{\bbp(\sum_i Y_i=n-2c_n-1)}\\
&&\leq \frac{C\sqrt{2}}{c}\sum_{k=A}^\infty\bbp(Y_1+\ldots+Y_{d(l)}=k)
\end{eqnarray*}
To sum up, if  $d(l)\leq c_n/2$,
\begin{equation}\label{dom}
\bbp(X_1+\ldots+X_{d(l)}\geq A)\leq 
\frac{C\sqrt{2}}{c}
\bbp(Y_1+\ldots+Y_{d(l)}\geq A)
\end{equation}
Recall that for every leaf $l$ of $S_n$, $d(l)\leq h_n$, and that 
because of \eqref{cvheight}, the condition $d(l)\leq c_n/2$ is satisfied 
for all leaves if $n$ is large enough. The bound using conditioning gave 
$$
\bbp\left(\frac{D(l)}{\sqrt{c_n}m_n}\geq (1+\varepsilon)H\right)
\leq 
\frac{(1+2t_n^2)^{h_n}\exp\left(-\frac{\varepsilon Ht_n\sqrt{c_n}}{4}\right)}
{\bbp(\sum_i Y_i=n-2c_n-1)}
$$
But using the stochastic domination bound \eqref{dom}, we can improve this to
$$
\bbp\left(\frac{D(l)}{\sqrt{c_n}m_n}\geq (1+\varepsilon)H\right)
\leq 
 \frac{C\sqrt{2}}{c}
(1+2t_n^2)^{h_n}\exp\left(-\frac{\varepsilon Ht_n\sqrt{c_n}}{4}\right)
$$
for $n$ large enough. Taking $t_n=c_n^{-1/4}$ and using \eqref{cvheight}, 
we find that the probability
$$
\bbp\left(\frac{D(l)}{\sqrt{c_n}m_n}\geq (1+\varepsilon)H\right)
$$
tends to 0 as $n$ goes to infinity, for every positive $\varepsilon$.
Likewise, if $e_n$ is a leaf in $S_n$ such that $d(l)= h_n$, one
can prove that the probability 
$$
\bbp\left(\frac{D(e_n)}{\sqrt{c_n}m_n}\leq (1-\varepsilon)H\right)
$$
goes to  0 as $n$ goes to infinity. 
This proves that 
$$
\frac{height(S_n)}{\sqrt{c_n}}
$$
converges in distribution to $H$. So we have the more general result

\begin{thm}
Let $(c_n, n\geq 1)$ be a sequence of integers such that
$c_n\to\infty$ as $n\to\infty$. Let $(T_n, n\geq 1)$ be a 
family of random trees such that for every $n\geq 1$, $T_n$ is a 
uniform Motzkin tree with $n$ vertices and $c_n+1$ leaves. Then 
$$
\frac{\sqrt{c_n}}{n}height(T_n)
$$ converges in distribution to the law of the height of a CRT. 
\end{thm}

\section{Conclusion}
In this paper, we gave two new samplers for rooted planar trees that satisfies a given partition of degrees. This sampler is now optimal in terms of random bit complexity.
We apply it to predict the average height of a random Motzkin in function of its frequency of unary nodes. We then prove some unconventional height phenomena (i.e. outside the universal $\Theta(\sqrt{n} )$ behaviour. Our work can certainly be extended to more complicate properties than the list of degrees. 
%We can expect similar samplers for simple patterns constraints, such as  fixed numbers of paths of fixed length, or fixed numbers of complete trees. 
Letters of a tree-alphabet could for instance encode more complicated patterns, whose number of leaves would be given by the function $f$.

\newpage

\bibliographystyle{alpha}
\bibliography{article}

\begin{thebibliography}{MTW04}

\bibitem[Ald93]{A93}
David Aldous.
\newblock The continuum random tree. iii.
\newblock {\em Ann. Probab.}, 21(1):248--289, 1993.

\bibitem[ARS97a]{ARS97}
Laurent Alonso, Jean-Luc Remy, and Ren{\'e} Schott.
\newblock A linear-time algorithm for the generation of trees.
\newblock {\em Algorithmica}, 17(2):162--183, 1997.

\bibitem[ARS97b]{ARS97-2}
Laurent Alonso, Jean-Luc Remy, and Ren{\'e} Schott.
\newblock Uniform generation of a schr{\"o}der tree.
\newblock {\em Inf. Process. Lett.}, 64(6):305--308, 1997.

\bibitem[BBJ13]{BBJ13}
Axel Bacher, Olivier Bodini, and Alice Jacquot.
\newblock Exact-size sampling for motzkin trees in linear time via boltzmann
  samplers and holonomic specification.
\newblock In {\em ANALCO}, pages 52--61, 2013.

\bibitem[BP10]{BP10}
Olivier Bodini and Yann Ponty.
\newblock {Multi-dimensional Boltzmann Sampling of Languages}.
\newblock In {\em {DMTCS Proceedings}}, number~01 in AM, pages 49--64, Vienne,
  Autriche, 2010.
\newblock 12pp.

\bibitem[CHar]{CH14}
N.~Curien and B.~Haas.
\newblock The stable trees are nested.
\newblock {\em Prob. Theory Rel. Fields}, to appear.

\bibitem[Dev86]{D86}
L.~Devroye.
\newblock {\em Non-uniform random variate generation}.
\newblock Springer-Verlag, 1986.

\bibitem[Dev12]{D12}
Luc Devroye.
\newblock Simulating size-constrained galton-watson trees.
\newblock {\em SIAM J. Comput.}, 41(1):1--11, 2012.

\bibitem[DFLS04]{DFLS04}
Philippe Duchon, Philippe Flajolet, Guy Louchard, and Gilles Schaeffer.
\newblock Boltzmann samplers for the random generation of combinatorial
  structures.
\newblock {\em Combinatorics, Probability {\&} Computing}, 13(4-5):577--625,
  2004.

\bibitem[DPT10]{DPT10}
Alain Denise, Yann Ponty, and Michel Termier.
\newblock {Controlled non uniform random generation of decomposable
  structures}.
\newblock {\em Theoretical Computer Science}, 411(40-42):3527--3552, 2010.

\bibitem[Duq]{D03}
T.~Duquesne.
\newblock A limit theorem for the contour process of conditioned galton-watson
  trees.

\bibitem[FS09]{FS09}
Philippe Flajolet and Robert Sedgewick.
\newblock {\em Analytic Combinatorics}.
\newblock Cambridge University Press, 2009.

\bibitem[FZC94]{FVCZ94}
Philippe Flajolet, Paul Zimmermann, and Bernard~Van Cutsem.
\newblock A calculus for the random generation of labelled combinatorial
  structures.
\newblock {\em Theor. Comput. Sci.}, 132(2):1--35, 1994.

\bibitem[GK98]{G98}
J.~Geiger and G.~Kersting.
\newblock The galton-watson tree conditioned on its height.
\newblock {\em Proceedings 7th Vilnius Conference.}, 1998.

\bibitem[Gne48]{G48}
B.~V. Gnedenko.
\newblock On a local limit theorem of the theory of probability.
\newblock {\em Uspehi Matem. Nauk (N. S.)}, 3(3(25)):187--194, 1948.

\bibitem[HH97]{H97}
Te~Sun Hao and M.~Hoshi.
\newblock Interval algorithm for random number generation.
\newblock {\em Information Theory, IEEE Transactions on}, 43(2):599--611, 1997.

\bibitem[KY76]{KY76}
Donald~E. Knuth and Andrew~C. Yao.
\newblock {The Complexity of Nonuniform Random Number Generation}.
\newblock In J.~F. Traub, editor, {\em Algorithms and Complexity: New
  Directions and Recent Results}. Academic Press, New York, 1976.

\bibitem[Mar08]{M08}
Ph. Marchal.
\newblock A note on the fragmentation of a stable tree.
\newblock {\em Discrete Math. Theor. Comput. Sci. Proc.}, pages 489--499, 2008.

\bibitem[MTW04]{MTW04}
George Marsaglia, Wai~Wan Tsang, and Jingbo Wang.
\newblock Fast generation of discrete random variables.
\newblock {\em Journal of Statistical Software}, 11(3):1--11, 7 2004.

\bibitem[Rem85]{R85}
Jean-Luc Remy.
\newblock Un proc{\'e}d{\'e} it{\'e}ratif de d{\'e}nombrement d'arbres binaires
  et son application a leur g{\'e}n{\'e}ration al{\'e}atoire.
\newblock {\em ITA}, 19(2):179--195, 1985.

\bibitem[Tut64]{T64}
W.~T. Tutte.
\newblock The number of planted plane trees with a given partition.
\newblock {\em The American Mathematical Monthly}, 71(3):pp. 272--277, 1964.

\bibitem[Vos91]{V91}
Michael~D. Vose.
\newblock A linear algorithm for generating random numbers with a given
  distribution.
\newblock {\em IEEE Transactions on Software Engineering}, 17(9):972--975,
  1991.

\bibitem[Wal77]{W77}
Alastair~J. Walker.
\newblock {An Efficient Method for Generating Discrete Random Variables with
  General Distributions}.
\newblock {\em ACM Transactions on Mathematical Software}, 3(3):253--256,
  September 1977.

\end{thebibliography}
\tikzstyle{lien}=[->,>=stealth,rounded corners=5pt,thick]
\tikzset{individu/.style={draw,thick,fill=#1!25},
         individu/.default={green}}
% dÃ©finition de lâ€™arbre
%% \begin{tikzpicture}
%%   \node[individu] (B) at (0,0) {};
%%   \node[individu=blue] (P) at (-3,2) {};
%%   \node[individu=red] (M) at (3,2) {};
%%   \node[individu=blue] (GPP) at (-4.5,4) {};
%%   \node[individu=red] (GMP) at (-1.5,4) {};
%%   \node[individu=blue] (GPM) at (1.5,4) {};
%%   \node[individu=red] (GMM) at (4.5,4) {};
%%   \draw[lien] (B) |- (-1,1) -| (P);
%%   \draw[lien] (B) |- (1,1) -| (M);
%%   \draw[lien] (P) |- (-4,3) -| (GPP);
%%   \draw[lien] (P) |- (-2,3) -| (GMP);
%%   \draw[lien] (M) |- (2,3) -| (GPM);
%%   \draw[lien] (M) |- (4,3) -| (GMM);
%% \end{tikzpicture}

%% \begin{tikzpicture}
%%  \node [individu] 
%%  \child { \node [individu=blue]{a}
%%           \child { node [individu=blue]{b} }
%%           \child { node [individu=red]{c} }
%%         }
%%   \child { \node [individu=red]{e}
%%           \child { node [individu=blue]{f}}
%%           \child { node [individu=red]{g} }
%%         };
%% \end{tikzpicture}

\end{document}